\documentclass[11pt]{article}
\usepackage{graphicx,amssymb,amsthm,amsmath,hyperref}
\usepackage{color}
\hypersetup{colorlinks=true}

\def\eee{\mathrm{e}}
\def\eps{\varepsilon}

\newcommand{\reals}{\mathbb{R}}

\def\ra{\rightarrow}

\def\<{\langle}
\def\>{\rangle}
\def\({\left(} 
\def\){\right)} 

\def\That{\widehat{T}}

\newtheorem{theorem}{Theorem}
\newtheorem{lemma}[theorem]{Lemma}

\newtheorem{claim}[theorem]{Claim}

\newcommand{\Q}{\mathcal{Q}}

\newcommand{\Tmix}{{\tau}_{\rm{mix}}}

\def\Prob#1{{\mathrm{Pr}\left({#1}\right)}}

\def\ProbCond#1#2{{\mathrm{Pr}\left({#1} \mid {#2} \right)}}

\def\dtv#1{d_{\mathrm TV}(#1)}

\newcommand{\Dvec}{\mathbf{D}}

\renewcommand{\vec}[1]{\mathbf{#1}}
\def\vt{\mathbf{t}}

\title{Fast Convergence of MCMC Algorithms for Phylogenetic
Reconstruction with Homogeneous Data on Closely Related Species}

\author{Daniel \v{S}tefankovi\v{c}\thanks{Department of Computer Science, University of Rochester,
Rochester, NY 14627.  Email: stefanko@cs.rochester.edu. Research
supported in part by NSF grant CCF-0910415.} \and Eric
Vigoda\thanks{College of Computing, Georgia Institute of
Technology, Atlanta GA 30332.  Email: vigoda@cc.gatech.edu.
Research supported in part by NSF grant CCF-0830298 and
CCF-0910584.} }

\date{November 22, 2010}

\begin{document}

\maketitle

\begin{abstract}
This paper studies a Markov chain for phylogenetic reconstruction
which uses a popular transition between tree topologies
known as subtree pruning-and-regrafting (SPR).  We analyze the
Markov chain in the simpler setting that the generating tree
consists of very short edge lengths, short enough so that each
sample from the generating tree
(or character in phylogenetic terminology) is likely to have only one mutation,
and that there enough samples so that the data looks like the generating
distribution.
We prove in this setting that the Markov chain is rapidly mixing, i.\,e.,
it quickly converges to its stationary distribution, which is the posterior
distribution over tree topologies.
Our proofs use that the leading term of the maximum likelihood function of
a tree $T$ is the maximum parsimony score, which is
the size of the minimum cut in $T$
needed to realize single edge cuts of the generating tree.
Our main contribution is a combinatorial proof that in our simplified setting,
SPR moves are guaranteed to converge quickly to the maximum parsimony tree.
Our results are in contrast to recent works
showing examples with heterogeneous data (namely, the data is
generated from a mixture distribution) where many natural Markov chains
are exponentially slow to converge to the stationary distribution.
\end{abstract}

\section{Introduction}

We study Markov Chain Monte Carlo (MCMC) methods for Bayesian inference
of phylogeny.   We begin by presenting the relevant background material
by defining phylogenetic trees,
evolutionary models (in Section \ref{sec:evolutionary-models}),
and the associated MCMC methods (in Section \ref{sec:SPR}).
We refer the interested reader to
Semple and Steel \cite{SS} for a more comprehensive introduction to the
mathematics of phylogeny.
Finally, we present our results and discuss related work in Section \ref{sec:our-results}.

A phylogenetic tree is an unrooted tree $T$ on $n$ leaves (called
taxa, corresponding to $n$ species) where internal vertices have
degree three. Let $E(T)$ denote the edges of $T$ and $V(T)$ denote
the vertices.
In the phylogenetic reconstruction problem, we observe a collection
of labelings of the leaves of $T$ from a set $\Omega$,
and our goal is to infer the tree $T$
from which they were generated from.  For example, if $\Omega=\{A,C,G,T\}$ then
we are given (aligned) DNA sequences for $n$ species, and we
are trying to determine the tree describing the evolutionary history of the
present-day species.

\subsection{Evolutionary models and Maximum Likelihood}
\label{sec:evolutionary-models}

The labelings on the leaves of $T$ are the projection of
labelings on all vertices of $T$, and these labelings of $V$ are generated in
the following manner.  There is a stochastic process along edges of $T$
(e.\,g., modeling the evolutionary process of DNA substitutions)
which is defined by a continuous-time Markov chain.  Thus, for
each edge $e\in T$ there is a $|\Omega|\times |\Omega|$ rate matrix $Q$
and a time $\vt_e>0$, which is called the branch length of $e$.
In this paper, as is typical in the phylogenetic setting,
we assume there is a single rate matrix $Q$ that is common to all
edges.  The rate
matrix is assumed to be reversible with respect to some
distribution $\pi$ on $\Omega$.   Hence, fix $\pi$ as the stationary
vector for $Q$ (i.\,e., $\pi Q = 0$). (The matrix $Q$ is usually
scaled so that we expect one ``substitution'' (i.\,e., change) per unit of time.)
The rate matrix $Q$ defines a continuous time Markov chain, and together
with $\vt_e$ defines a transition matrix on edge~$e$:
\begin{equation}
\label{def:P}
  P_e=\exp(\vt_eQ) = I + \vt_eQ + \vt_e^2Q^2/2! + \vt_e^3Q^3/3! + \dots
\end{equation}
The matrix $P_e$ is a stochastic matrix of size
$|\Omega|\times|\Omega|$, and thus defines a discrete-time Markov
chain, which is time-reversible with stationary distribution
$\pi$, i.\,e., $\pi P_e = \pi$, and
$\pi_i (P_e)_{ij} = \pi_j (P_e)_{ji}$ (for all $i,j\in\Omega$).

The simplest four-state (i.\,e., $|\Omega|=4$) evolutionary
model has a single parameter for the
off-diagonal entries of the rate matrix $Q$; this model is known as
the Jukes-Cantor model.  The most general reversible four-state
model is the GTR (general time reversible) model. For
$|\Omega|=2$ (often studied for mathematical interest),
the model is binary and the rate matrix has a single
parameter; this model is known as the CFN (Cavender-Farris-Neyman)
model.  See Felsenstein \cite{F:book} or Yang \cite{Yang}
for an introduction to these evolutionary models.

Given $T$, the rate matrix $Q$ and the branch lengths
$\vt=(\vt_e)_{e\in E(T)}$, we then define the
following distribution on labelings of the vertices of $T$. Let
$P_e=\exp(\vt_eQ)$ for $e\in E(T)$. We first orient the edges of
$T$ away from an arbitrarily chosen root $r$ of the tree. (We can
choose the root arbitrarily since each $P_e$ is reversible with
respect to $\pi$.) Then, the probability of a labeling
$\ell:V(T)\ra\Omega$ is
\begin{equation}
\label{eq:defp}
\mu'_{T,Q,\vt}(\ell):=\pi(\ell(r))\prod_{\overrightarrow{uv}\in E(T)}
P_{uv}(\ell(u),\ell(v)).
\end{equation}
The distribution $\mu'_{T,Q,\vt}$ can be generated in an
equivalent algorithmic manner. Choose $\ell(r)$ from $\pi$.  Then
for each edge $e=(u,v)\in E(T)$, given an assignment for exactly
one of the endpoints, say $\ell(u)$, choose $\ell(v)$ from the
distribution defined by the row of $P_e$ corresponding to the
label $\ell(u)$.

Let $\mu_{T,Q,\vt}$ be the marginal distribution of $\mu'_{T,Q,\vt}$ on
the labelings of the leaves of $T$ (thus $\mu_{T,Q,\vt}$ is a distribution
on $\Omega^n$ where $n$ is the number of leaves of $T$). Fix $T^*$ with
parameters $Q^*$ and $\vt^*$ as the generating tree. The goal of
phylogeny reconstruction is to reconstruct $T^*$
(and possibly $Q^*$ and $\vt^*$)
from $\mu_{T^*,Q^*,\vt^*}$ (more precisely, from independent samples from $\mu_{T^*,Q^*,\vt^*}$).

Let $\Q$ denote a set of rate matrices with non-zero entries where
$Q^*\in \Q$. The set $\Q$ is the set of possible rate matrices.
The set can be arbitrary, usually it is determined by the model
considered (e.\,g., for the Jukes-Cantor model $\Q$ would contain
rate matrices whose off-diagonal entries are the same). One often
assumes that the rate matrix $Q^*$ is known.  In this case we
would set $\Q=\{Q^*\}$.  On the other hand, our results also apply
if one sets $\Q$ to be the set of all rate matrices with non-zero
entries.

We consider the likelihood of a tree $T$ as, the maximum over rate matrices $Q\in\Q$ and
over assignments of non-zero
branch lengths $\vt$ to the edges of $T$, of the probability that the
tree $(T,Q,\vt)$ generated $\mu$. More formally, the maximum expected
log-likelihood of tree $T$ for distribution $\mu^*$ is defined by
\begin{equation}\label{e:aa1}
{\cal L}_{T}(\mu^*) = \sup_{Q\in\Q} \sup_{\vt} {\cal L}_{T,Q,\vt}(\mu^*),
\end{equation}
where
\begin{equation}\label{e:aa2}
 {\cal L}_{T,Q,\vt}(\mu^*) = \sum_{y\in\Omega^n} \mu^*(y)\ln(\mu_{T,Q,\vt}(y)).
\end{equation}

For a set of characters $\mathbf{D}=(D_1,\dots,D_N)$
where $D_i\in\Omega^n$, define the log-likelihood of a tree $T$ as
\begin{eqnarray*}
{\cal L}_{T}(\mathbf{D}) & =  &
\sup_{Q\in\Q} \sup_{\vt}\,\, \ln\left({\mu}_{T,Q,\vt}\left(\mathbf{D}\right)\right)
\\ & = &
\sup_{Q\in\Q} \sup_{\vt}\,\, \sum_{i=1}^N \ln\left(\mu_{T,Q,\vt}(D_i)\right).
\end{eqnarray*}

Our goal is to sample from the distribution on the
set of phylogenetic trees with $n$ leaves
where the weight of a tree is ${\cal L}_{T}(\mathbf{D})$.
In Section \ref{sec:posterior} we will look at the
straightforward extension to the setting where we are
given a prior on trees and parameters $Q,\vt$, and our goal is
to sample from the posterior distribution.

\subsection{SPR Markov Chain}

\label{sec:SPR}

We analyze a Markov chain using transitions made by
subtree pruning-and-regrafting (SPR).  SPR transitions are
a natural combinatorial transition, which
is also popular in practice.
In Section
\ref{remarks} we discuss several other well-studied
choices for the transitions.  Here we consider trees weighted
by their maximum likelihood.  In Section \ref{sec:posterior}
we discuss how the Markov chain definition and our main
result extends to sampling the posterior distribution.

An SPR transition from a tree T works by choosing an (internal or terminal)
edge $e=(u,v)$.  If $e$ is an internal edge we consider one of the two subtrees
in $T\setminus e$, either the subtree $S_u$ containing  $u$
or the subtree $S_v$ containing $v$.
Let $S_u$ denote the selected subtree.
If $e$ is a terminal edge, let $S_u$ be the endpoint of $e$ that is a leaf.
Let $T'$ denote the tree formed by removing $S_u$ from $T$, in
particular, we remove $S_u$ and edge $e$ from $T$ and ``smooth away'' the
vertex $v$ (that is, contract one of the two adjacent edges).
We then choose an edge $e^*$ in $T'$ and
we attach $S$ onto $e^*$ by adding a new intermediate
vertex along $e^*$.  See Figure \ref{fig:SPR} for an illustration.
Let $SPR(T,S_u,e^*)$ denote the tree resulting from the above transition.

\begin{figure}[htb]
\begin{center}
\includegraphics[type=pdf,ext=.pdf,read=.pdf,height=2in]{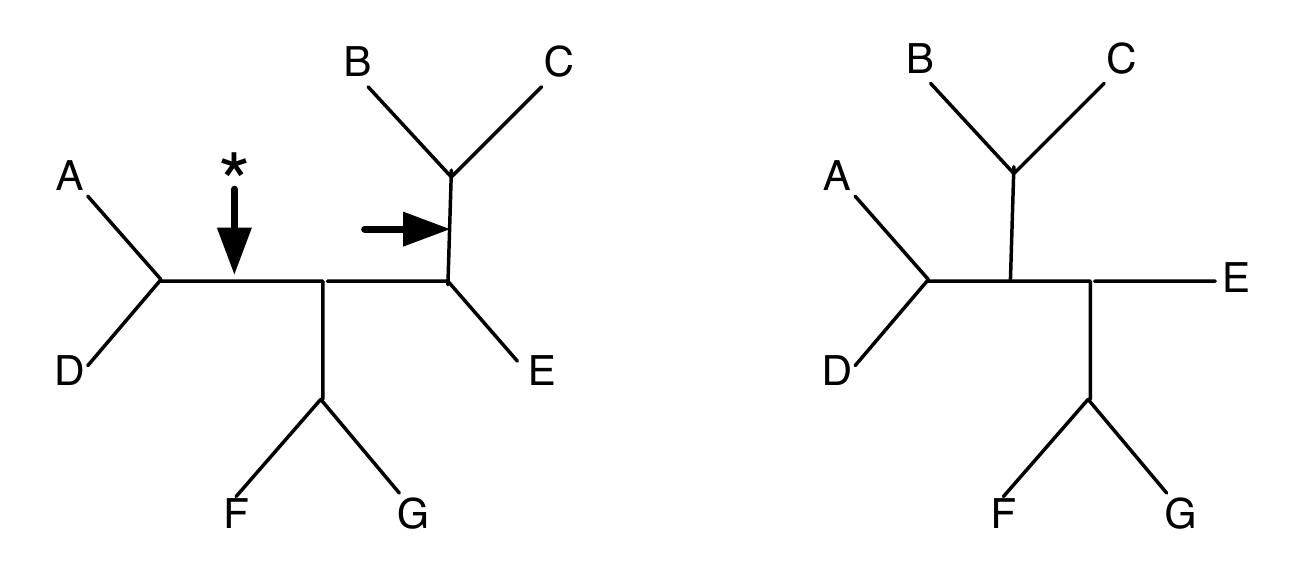}
\caption{Illustration of an SPR transition.
The randomly chosen edge $e$ is marked by an arrow.  The subtree containing
$B$ and $C$ is pruned, and then regrafted at the edge $e^*$
marked by a starred arrow.  The
resulting tree is illustrated.}
\label{fig:SPR}
\end{center}
\end{figure}

We analyze the following Markov chain
which chooses a random subtree $S$
to prune, and then chooses an edge to regraft $S$ along based on the
maximum likelihood of the resulting tree.  This Markov chain
is an analogous to heat bath chains studied in Statistical Physics
(as opposed to Metropolis chains), e.g., see \cite{KLW}, thus we refer
to the below chain as the Heat Bath SPR Markov Chain.
Here is the formal
definition of the transitions $T_t\rightarrow T_{t+1}$ of the
Heat Bath SPR Markov Chain.

From a tree $T_t$ at time $t$ we proceed as follows.
\begin{enumerate}
\item
Choose a random subtree $S$ of $T_t$, by choosing a
random edge $e$ and then choosing one of the two subtrees
hanging off of $e$.
Let $T'$ denote the tree formed by deleting
$S$ and $e$ from $T$.
\item
For each edge $e^*$ of tree $T'$,
let $w(e^*) = \mathcal{L}_{\That}(\vec{D})$ where $\That=SPR(T,S,e^*)$
is the tree formed by pruning $S$ from $T$ and regrafting $S$
onto edge $e^*$.
 Let $\omega$ be the distribution on edges of $T'$ where
$\omega(e^*) = w(e^*)/Z$
and $Z=\sum_{e'\in E(T')} w(e')$.
\item Sample an edge $e^*$ from the distribution $\omega$ on edges of $T'$.
\item Graft $S$ onto edge $e^*$ and move to this new tree, i.\,e.,
set $T_{t+1} = SPR(T,S,e^*)$.
\end{enumerate}

We now verify that the above Markov chain is ergodic and reversible
with respect to the distribution $\pi$ on trees where $\pi(T)\propto \mathcal{L}_T(\vec{D})$,
and thus $\pi$ is also the unique stationary distribution. Let $T_1$ and $T_2$ be
neighboring states of the Markov chain. Let $S$ be the tree that is pruned and
regrafted to obtain $T_2$ from $T_1$. Note that the same tree $S$ can be pruned
and regrafted to obtain $T_1$ from $T_2$. The transition probability from $T_1$
to $T_2$ is the probability of choosing $S$ in step $1$ times
$\mathcal{L}_{T_2}(\vec{D})/Z$. Similarly the transition probability from $T_2$
to $T_1$ is the probability of choosing $S$ in step $1$ times
$\mathcal{L}_{T_1}(\vec{D})/Z$ (note that $Z$ is the same in both cases since pruning
$S$ results in the same tree $T'$). The detailed balance condition
is satisfied for $\pi(T)\propto \mathcal{L}_T(\vec{D})$ and hence it
is the unique stationary distribution.

Let $\dtv{\mu,\nu}$ denote the (total) variation distance
between
a pair of probability distributions $\mu$ and $\nu$ defined on the
same finite, discrete space, and let $P^t(T_0,\cdot)$ denote the
$t$-step distribution of the Markov chain from initial state $T_0$.
The mixing time $\Tmix$ is defined as
\[  \Tmix = \max_{T_0} \min\{t: \dtv{P^t(T_0,\cdot),\pi} \leq 1/2\eee\},\]
which is the time to reach variation distance $\leq 1/2\eee$ of the stationary distribution
from the worst initial state.  Note, it is straightforward to then
``boost'' so that
for any $\delta>0$,
after $\Tmix\ln(1/\delta)$ steps we are within variation distance
$\le\delta$ of $\pi$, from the worst initial state
(see Aldous \cite{Aldous}).

\subsection{Results on MCMC for Phylogenetic Reconstruction}
\label{sec:our-results}

MCMC algorithms are an important tool for phylogenetic reconstruction.
MrBayes
\cite{mrbayes}
 is a popular program that relies on MCMC methods for
 Bayesian inference of phylogeny.  MrBayes
 uses a sophisticated variant of MCMC known
as Metropolis-Coupled MCMC
\cite{MCMCMC}.

For statistical inference problems, such as phylogenetic reconstruction,
it is often easy to design appropriate MCMC algorithms,
such as the above Markov chain we defined using SPR transitions,
which converge in the limit over
time to the desired posterior distribution.
However, the computational efficiency of these methods
rely on their {\em fast} convergence to the posterior distribution.  Since
theoretical results are typically lacking, heuristic methods are used to measure
convergence to the desired distribution.  Hence, there are often no rigorous guarantees
on the scientific computations which rely on the random samples produced by
the MCMC methods.  Our goal is to provide some theoretical understanding of
settings where MCMC methods for phylogenetic reconstruction are
provably fast and hence yield
accurate results, and settings where the MCMC methods are
slow and consequently the samples may be misleading.

There are several works with computational experiments on the
convergence rates of MCMC algorithms for phylogenetic
reconstruction, e.\,g., see the recent works \cite{BKHR,LMHLR}.
There is relatively little theoretical work.
Diaconis and Holmes
\cite{Diaconis}
proved fast convergence of a Markov
chain to the uniform distribution over phylogenetic trees.
Recently, several works have shown examples of heterogeneous data
where MCMC algorithms are provably slow to converge. Mossel and
Vigoda
\cite{MV1,MV2}
proved slow convergence for a class of
examples with data arising from a uniform mixture of a pair of
5-taxa trees (with different topologies).  \v{S}tefankovi\v{c} and
Vigoda
\cite{SV1,SV2}
proved slow convergence for a class of
mixture examples from a pair of 5-taxa trees that share the same
topology but differ in their branch lengths. In these slow mixing
results, the convergence time is exponential in the number of
characters (i.\,e., sequence length).

In this paper we show fast convergence for data from a homogenous
source of closely related species.  In particular, for data
generated from a single tree (of any size) when all the branch
lengths are sufficiently short, we prove fast convergence.  The
requirement of sufficiently short branches is for our proof
technique, but it is important to note that the slow mixing
results mentioned earlier
\cite{MV1,MV2,SV1,SV2}
require, in an
analogous manner, sufficiently short branches.
If one searches for the tree with the maximum likelihood (or maximum
a posteriori probability) our methods show that in our setting
of very short branch lengths,
the space of trees (connected by the SPR moves) has no local maxima and hence
one can find the optimal tree using hill climbing.

For simplicity, we present our results here for the case where the weight of a tree is
the maximum likelihood of generating the given data $D$ where the
maximum is over a rate matrix $Q$ (common to all edges) and a set of
branch lengths $\vt$.  This is closely related to the posterior distribution
when the priors are uniformly distributed.  Our results extend to $\delta$-regular
priors, which are priors that
are  lower bounded by some $\delta>0$, see Section \ref{sec:posterior}
 for a discussion
on the extension of our results to sampling the posterior distribution.
We are interested in the mixing time $\Tmix$,
defined as the number of steps until the chain is within variation distance $\leq 1/2\eee$
of the stationary distribution.

We prove that the Heat Bath SPR Markov Chain
converges quickly to its stationary distribution when the
data is generated from a tree $T^*$ where all of the branch lengths
are sufficiently small, and there are sufficiently many samples generated from $T^*$.
 Here is the formal statement of our main result.

\begin{theorem}
\label{thm:main}
Consider any reversible $4$-state model,
any phylogenetic tree $T^*$ on $n$ taxa and
any rate matrix $Q^*$ with no zero entries.
For all $\alpha_{\min}>0$, there exists $\epsilon_0>0$,
such that
for all $0<\eps<\eps_0$ and
any choice of branch lengths $\vt^*_e\in\(\alpha_{\min}\epsilon,\epsilon\)$ for $e\in E(T^*)$, there exists $N_0>0$
where the following holds.

For a data set with $N>N_0$ characters, each
chosen independently from the distribution $\mu_{T^*,Q^*,\vt^*}$,
then, with probability $>1-\exp(-\sqrt{N})$ over the data generated,
the Heat Bath SPR Markov Chain
has mixing time $\leq 50n$.
\end{theorem}

Since there is considerable quantification in the above
theorem we will take a moment to dissect it at a high-level.
First off, the requirement that $N>N_0$ comes from needing the
data to look very much like the generating distribution $\mu^*=\mu_{T^*,Q^*,\vt^*}$.
Therefore, how much data we need depends on several quantities, such as
the minimum probability of a configuration arising in $\mu^*$, which depends
on the minimum branch length and the minimum rate in $Q^*$.  Hence,
$N_0$ depends on $\alpha_{\min}$ and $Q^*$, and is exponential in $n$.
A somewhat related question has been studied by Steel and Sz\'{e}kely
\cite{SteelSzekely1,SteelSzekely2,SteelSzekely3} on how large $N$
needs to be so that the maximum likelihood tree is the generating tree $T^*$.
In their results one also needs $N$ to be exponential in $n$.

Our proof uses the fact that in the setting of Theorem \ref{thm:main}, where the
branch lengths are sufficiently short, the leading term of the maximum likelihood function
is actually maximum parsimony.  Such a result is well-known in the mathematical phylogeny
community, and was first observed by Felsenstein \cite{F:parsimony}.  We require
a more detailed statement of such a result, which we present in
Lemma \ref{lem:opt-weights} in Section \ref{sec:likelihood}.
Our main technical contribution is a combinatorial proof that, in the setting of Theorem \ref{thm:main},
SPR moves can be used in a greedy manner to quickly find the maximum parsimony tree.
This result is presented in Section \ref{sec:A-analysis}.
Finally, in Section \ref{sec:rapid-mixing} we show how Theorem \ref{thm:main}
follows in a straightforward manner from these combinatorial results.
In Section \ref{sec:posterior} we
discuss how Theorem~\ref{thm:main} extends to Bayesian inference.
We make some concluding remarks in Section \ref{remarks}.

\section{Proof of Rapid Mixing}
\label{sec:proof}

\subsection{Overview}

To prove Theorem \ref{thm:main} we will analyze, for every tree $T$,
the maximum expected
log-likelihood ${\mathcal L}_{T}(\mu^*)$ where
$\mu^*=\mu_{T^*,Q^*,\vt^*}$ (recall that ${\mathcal L}_{T}(\mu^*)$ is the maximum
expected log-likelihood of $T$ maximized over all rate matrices $Q$ and all
edge lengths $\vt$, see \eqref{e:aa1}).
To analyze ${\mathcal L}_{T}(\mu^*)$
we will consider the dominant terms of the likelihood function.
We will show that
\[  {\mathcal L}_{T}(\mu^*) =
{\cal E}(\pi^*) + A(T)\eps\ln\eps + o(\eps\ln\eps),
\]
where ${\cal E}(\pi^*)$ is the entropy of the stationary distribution of $Q^*$
and thus is the same for every $T$.  By taking $\eps$
sufficiently small the last term $o(\eps\ln\eps)$ can be ignored.
Therefore, the dominant term is $A(T)\eps\ln\eps$.  We will prove that the
function $A(T)$ decreases with each optimal SPR move.
Hence, since
$\ln\eps$ is negative, we then have that $T^*$ has the highest
maximum expected log-likelihood, and as the Markov chain gets closer
to $T^*$ the maximum expected log-likelihood will increase.
Theorem \ref{thm:main} will then follow in a straightforward manner.

\subsection{Analyzing Likelihood}
\label{sec:likelihood}

Let $T$ be a tree with leaves $\{1,\dots,n\}$. Let $R$ be a
partition of the leaves into two parts $(R_1,R_2)$. Note
that we only consider the partition of the leaves without
any regard for the internal vertices.
Let ${\mathrm cut}_R(T)$ denote the size of the cut (i.\,e., a subset of edges)
of minimum size that disconnects $R_1$ from~$R_2$.
Tuffley and Steel \cite{TuffleySteel} showed that the quantity
${\mathrm cut}_R(T)$ is the parsimony score of the character corresponding
to $R$ (see also Semple and Steel \cite[Proposition 5.1.6]{SS}), where
the character corresponding to $R=(R_1,R_2)$
assigns all leaves in $R_1$ some $\alpha\in\Omega$ and assigns
all leaves in $R_2$ some $\beta\in\Omega,\beta\neq\alpha$.

For an edge $e\in T$, the removal of $e$ splits $T$ into two components.
This induces a partition of the leaves of $T$ into two parts.
We will call this partition $R(T,e)$.

The following is the main technical tool for our results.  The lemma
describes the high-order terms of the likelihood function as
$\eps\rightarrow 0$. Throughout this paper, the asymptotic
notations $o()$ and $O()$ are parameterized by $\eps\rightarrow
0$.

Roughly speaking, the lemma shows there is a function $A(T)$ which
plays a leading role in the maximum likelihood.
In words, $A(T)$ looks at each partition $R$ of leaves in
the generating tree $T^*$ realized by
cutting a single edge $e^*$.  It then considers the minimum number of edges
in $T$ to realize this partition $R$ times the branch length of $e^*$ in the
generating tree.
As mentioned in the introduction, there are earlier results which show
that the leading term of the likelihood function is the parsimony score, e.g.,
Felsenstein~\cite{F:parsimony}, and in the below
lemma, the function $A$ is the leading term of the expected parsimony
score.
 We require a more detailed
statement than we found in the literature.

\begin{lemma}
\label{lem:opt-weights} Let $T^*,T$ be trees with leaves
$\{1,\dots,n\}$ and $Q^*$ be a rate matrix reversible with respect
to $\pi^*$. Assume that the matrix $Q^*$ is normalized (that is,
$\sum_{i\neq j}\pi^*_i Q^*_{ij} = 1$) and that $Q^*$ has no zero
entries. Let $T^*$ have branch lengths $\vt_e^*=\alpha^*_e\eps$,
where $\alpha^*_e\in [\alpha^*_{\min},\alpha^*_{\max}]$, for all
$e\in E(T^*)$, where $\alpha_{\min}>0$. Let
\begin{equation}
\label{def:A}
A=A_{T^*,\vec{\alpha^*}}(T) =\sum_{e\in E(T^*)} \alpha^*_e\,
{\rm cut}_{R(T^*,e)}(T).
\end{equation}
For $\mu^*=\mu_{T^*,Q^*,\vt^*}$ the following holds:
\begin{equation}\label{eee2}
{\mathcal L}_{T}(\mu^*) =  {\cal E}(\pi^*) + A\eps\ln\eps + O(\eps\ln\ln(1/\eps)),
\end{equation}
where ${\cal E}(\pi^*)=\sum_{i\in\Omega} \pi^*_i\ln\pi^*_i$
is the entropy of $\pi^*$,
and the constant in the $O(\cdot)$ is independent of the choice of
the $\alpha^*_e$ (but does depend on
$\alpha^*_{\min}$ and $\alpha^*_{\max}$).
\end{lemma}

As a consequence of the above lemma, to analyze the
expected log-likelihood on the tree space, when $\eps$
is sufficiently small, we simply have to consider the function $A=A(T)$.
In the next subsection we will investigate the combinatorial
properties of $A$.

Before presenting the proof of Lemma \ref{lem:opt-weights} we give a brief outline of its proof.
Let $\mu^*$ denote the probability distribution defined by $Q^*$ and $T^*$
on assignments of labels from $\Omega$ to the leaves.
In this generating distribution $\mu^*$,
to prove \eqref{eee2}
we only need to consider two types of assignments.
The first type are
constant assignments where no substitutions occur and
thus all leaves receive the
same label $i\in\Omega$, these are denoted as $\sigma_i$.
The second type are
 assignments obtained by a substitution along just one
edge $e^*$.  In this case, the cut obtained by deleting edge $e^*$ plays
an important role.  By deleting $e^*$ from $T^*$, the leaves
are partitioned into two sets $R_1$ and $R_2$, denoted as $R(T^*,e^*)$.
If a substitution
only occurs along edge $e^*$, then the leaves in $R_1$ will receive
the same label $i\in\Omega$, and the leaves in $R_2$ will receive another label
$j\in\Omega, j\neq i$.  We denote such an assignment by $\sigma^{e^*}_{ij}$.
Any other type of assignment requires at least two substitutions,
and hence has probability at most $O(\eps^2)$, which is dominated
by the $O(\eps\ln\ln(1/\eps))$ term of~\eqref{eee2}.

For any tree $T$, to prove that ${\mathcal L}_{T}(\mu^*)$ is lower
bounded by the right-hand side of \eqref{eee2}, we compute the
expected log-likelihood of $\mu^*$ for the rate matrix $Q=Q^*$ and
the set of branch lengths $\vt$ where $\vt_e=\eps$
for every edge $e$.  For each edge $e^*$
and its corresponding assignment $\sigma^{e^*}_{ij}$, the quantity
${\rm cut}_{R(T^*,e)}(T)$ is the minimum number of edges which
require a substitution to obtain the assignment
$\sigma^{e^*}_{ij}$ on $T$. Hence, the quantity $A=A(T)$ plays an
important role when we sum over all edges $e^*$ of $T^*$.  In
particular, by a calculation (as detailed in \eqref{lower-final} below), the expected log-likelihood ${\mathcal
L}_{T,Q,\vt}(\mu^*)$, for this
set of branch lengths $\vt$,  is  $\sum_{i\in\Omega} \pi_i^*\ln\pi^*_i +
 A\eps\ln\eps + O(\eps)$.  Since $O(\eps)=O(\eps\ln\ln(1/\eps))$,
 this implies the lower bound of~\eqref{eee2}.

 To obtain the upper bound of \eqref{eee2} we consider three cases:
 when the rate matrix $Q$ has a stationary distribution different from
 $Q^*$, when there is an edge $e$ where $\vt_e$ is long
 (namely $\geq \eps(\ln(1/\eps))^2$),
 and when all edges are short.  In the first case of different stationary
 distributions, by considering the constant assignments it will
 be easy to establish that there is a difference in the first
 term of the right-hand side of~\eqref{eee2}.
When there is a long edge, then the constant assignments are too
unlikely to occur.  Finally, if all edges are shorter than
$\eps(\ln(1/\eps))^2$, then, by calculation, we show that
the expected log-likelihood is at most the
right-hand side of~\eqref{eee2}.

We now present the formal proof of Lemma \ref{lem:opt-weights}.

\begin{proof}[Proof of Lemma \ref{lem:opt-weights}]
We first make some observations about the distribution~$\mu^*$.
Let $P^*$ denote the transition matrix for $Q^*$, as defined in \eqref{def:P}.

Note, for any $e$, any $i,j\in\Omega$ where $i\neq j$,
we have
\begin{equation}
\label{P:ij}
(P^*_e)_{ij} = Q^*_{ij}\alpha^*_e\eps + O(\eps^2).
\end{equation}
For any $i\in\Omega$ we have
\[
(P_e^*)_{ii} = 1 - \sum_{j\neq i} (P^*_e)_{ij} = 1 + O(\eps).
\]

For $i\in\Omega$, let $\sigma_i\in\Omega^n$ denote the
constant assignment $\sigma_i(v)=i$ for all leaves~$v$.
Note to achieve $\sigma_i$ in $\mu^*$ we assign label $i$ to the root and
then we have no substitutions, or we have at least two edges with substitutions.
Thus,
\begin{equation}\label{er2}
\mu^*(\sigma_i) = \pi^*_i\prod_{e\in E(T^*)} (P_e^*)_{ii}  + O(\eps^2) = \pi^*_i + O(\eps).
\end{equation}
For an edge $e\in E(T^*)$ and $i,j\in\Omega$ where $i\neq j$,
let $\sigma^e_{ij}\in\Omega^n$ denote the assignment of label $i$ to
all leaves in one of the partitions of $R(T^*,e)$ and label $j$ to all
leaves in the other partition of $R(T^*,e)$.  In this case we have:
\begin{equation}\label{er}
\mu^*(\sigma^e_{ij}) = \pi^*_iQ^*_{ij}\alpha^*_e\eps + O(\eps^2).
\end{equation}
(To see why \eqref{er} is correct, w.l.o.g., assume that the root is a leaf
in the first partition of $R(T^*,e)$, and hence to achieve $\sigma_{ij}$
we need to label the root by $i$ and have a substitution on $e$, or
at least two edges with substitutions.)

Now we compute ${\cal L}_{T,Q,\vt}(\mu^*)$ where $\vt_e=\eps$ for
each edge $e$ of $T$ and $Q=Q^*$. Again we will make some observations
about $\mu=\mu_{T,Q,\vt}$. By the same reasoning as we used for \eqref{er2},
we obtain
\begin{equation}\label{er3}
\mu(\sigma_i) = \pi^*_i + O(\eps).
\end{equation}
We can obtain assignment $\sigma^e_{ij}$ on $T$ using a substitution on
$ {\rm cut}_{R(T^*,e)}(T)$ edges, and we cannot obtain this assignment
with fewer substitutions. Hence,
\begin{equation}\label{er4}
\mu(\sigma^e_{ij}) = \Theta(\eps^{{\rm cut}_{R(T^*,e)}(T))}).
\end{equation}
Therefore,
\begin{equation}\label{er5}
\ln\mu(\sigma^e_{ij}) = \Theta(1) +  {\rm cut}_{R(T^*,e)}(T)\ln\eps.
\end{equation}
In order to compute the high-order terms of ${\cal L}_{T,Q,\vt}(\mu^*)$
we do not need to consider labelings other than $\sigma_i$ and $\sigma^e_{ij}$
(the other labelings have probability $O(\eps^2)$ in $\mu^*$).

Combining \eqref{er}, \eqref{er2}, \eqref{er3}, and \eqref{er5} we obtain
\begin{eqnarray}
{\mathcal L}_{T,Q,\vt}(\mu^*)
& = &
O(\eps^2\ln\eps) + \sum_{i\in\Omega} (\pi^*_i+O(\eps)) \ln (\pi^*_i+O(\eps))
\nonumber
\\ &&
+
\sum_{e\in E(T^*)}\sum_{i\neq j} (\pi^*_iQ^*_{ij}\alpha^*_{e}\eps + O(\eps^2))(\Theta(1) +
{\rm cut}_{R(T^*,e)}(T)\ln\eps)
\nonumber
\\
& = &
 O(\eps) +
\sum_{i\in\Omega} \pi_i^*\ln\pi^*_i +
 A\eps\ln\eps,
 \label{lower-final}
\end{eqnarray}
where in the last inequality we used the fact that $Q^*$ is normalized.
This proves
the lower bound in \eqref{eee2}.

It remains to prove the upper bound in \eqref{eee2}.  We will show
that no rate matrix and no assignment of branch lengths can do
better than the bound established in~\eqref{lower-final}. Let $Q$
be a rate matrix with stationary distribution $\pi$. If
$\pi\neq\pi^*$ then we bound ${\mathcal L}_{T,Q,\vt}(\mu^*)$ as
follows. First, note that the terms in the sum \eqref{e:aa2} are
negative and hence to obtain an upper bound we will only consider
the constant assignments. Second, the probability of constant
assignment $\sigma_i$ in $\mu^*$ is $\mu^*(\sigma_i)\leq \pi_i^*$
and similarly $\mu(\sigma_i)\leq\pi$.  Thus
$$
{\mathcal L}_{T,Q,\vt}(\mu^*) \leq \sum_{i\in\Omega} \pi_i^*\ln\pi_i
= \sum_{i\in\Omega} \pi_i^*\ln\pi_i^* - D_{KL}(\pi^*\|\pi),
$$
where $D_{KL}(\pi^*||\pi) := \sum_{i\in\Omega} \pi^*_i(\ln(\pi^*_i/\pi_i))$
is the K-L divergence of $\pi$ from $\pi^*$.  Since,
by the Gibbs' inequality, the KL-divergence is positive when $\pi\neq \pi^*$,
we have established the upper bound in \eqref{eee2} for the case $\pi\neq\pi^*$.

Now we assume $\pi=\pi^*$. Let $\vt$ be an assignment of branch lengths
to $T$. Let $\mu = \mu_{T,Q,\vt}$. Suppose that there exists an edge $f\in E(T)$
with branch length $\vt_f > \eps(\ln(1/\eps))^2$.  We are going
to show that such a $\vt$ has a tiny log-likelihood because of the
constant leaf labelings (i.\,e., $\sigma_i, i\in\Omega$).
By \eqref{def:P}, we have
$(P_f)_{ii} \leq 1 - q_{\min}\eps(\ln(1/\eps))^2 + O(\eps^2(\ln(1/\eps))^4)$,
where $q_{\min} = \min_{i,j\in\Omega} |Q(i,j)|$.  Hence,
$$
\mu(\sigma_i)\leq \pi_i\left(1 - q_{\min}\eps(\ln(1/\eps))^2 + O(\eps^2(\ln(1/\eps))^4)\right).
$$
Thus
\begin{eqnarray}\label{er6}
{\mathcal L}_{T,Q,\vt}(\mu^*)
& \leq &
 O(\eps) + \sum_{i\in\Omega} \pi^*_i \(\ln (\pi_i) -q_{\min}\eps(\ln(1/\eps))^2 +  O(\eps^2(\ln(1/\eps))^4)\)
 \nonumber
 \\
 & \leq &
{\cal E}(\pi^*)-q_{\min}\eps(\ln(1/\eps))^2+O(\eps).
\end{eqnarray}
As $\eps\rightarrow 0$, \eqref{er6} is smaller than the right-hand side of \eqref{eee2}
and we are done.

We are now left with the case in which all edges $f\in E(T)$ have branch lengths $\vt_f\leq\eps(\ln(1/\eps))^2$.
Since we can generate the leaf labelings starting
from any vertex, then by starting at a leaf, we see that:
\begin{equation}\label{er78}
\ln\mu(\sigma_{i}) \leq \ln \pi_i.
\end{equation}
Moreover, for $e\in E(T^*)$,  to generate $\sigma^e_{ij}$, we need to
have substitutions across
all edges in a cut that realizes $R(T^*,e)$.  Since the edges are short this
happens with probability $\le \left(\eps(\ln(1/\eps))^2\right)^k$ where
$k$ is the size of the cut. Since there are at most $2^n$ such cuts and
each has size at least ${\rm cut}_{R(T^*,e)}(T)$, we have that:
\begin{equation}\label{er7}
\ln\mu(\sigma^e_{ij}) =  {\rm cut}_{R(T^*,e)}(T)(O(\ln\ln(1/\eps))+\ln\eps).
\end{equation}
Hence,
\begin{eqnarray*}
{\mathcal L}_{T,Q,\vt}(\mu^*)
& \leq & O(\eps^2\ln\eps) + {\cal E}(\pi)
\\
& & +   \sum_{e\in E(T^*)}\sum_{i\neq j} (\pi^*_iQ^*_{ij}\alpha^*_{e}\eps + O(\eps^2))
 {\rm cut}_{R(T^*,e)}(T)(O(\ln\ln(1/\eps))+\ln\eps)
\\ &= &
 O(\eps\ln\ln(1/\eps)) +
{\cal E}(\pi) +
 A\eps\ln\eps.
 \end{eqnarray*}
\end{proof}


\subsection{Analyzing the Cut-Distance $A(T)$}
\label{sec:A-analysis}

\newcommand{\GOOD}{\mathrm{GOOD}}
\newcommand{\BAD}{\mathrm{BAD}}
\newcommand{\Tmin}{T_{\min}}

In light of Lemma \ref{lem:opt-weights}, we need to analyze how
$A(T)$ changes with SPR moves.  By taking $N$ sufficiently large,
for each subtree $S$,
we will only need to analyze the effect of
the optimal SPR move for $S$ (optimal in terms of minimizing $A(T')$).

The quantity $A(T)$ looks at cuts obtained by single edges of $T^*$.
For a tree $T$, we classify the edges of $T^*$ as
good or bad if their corresponding cut in $T^*$ is
realizable in $T$ by cutting a single edge.  More precisely, let
\[
\GOOD_{T^*}(T) = \{e^*\in E(T^*): \mbox{there exists } e\in E(T)
\mbox{ where } R(T,e) = R(T^*,e^*)\},
\]
be the set of good edges for $T$.  Let $\BAD_{T^*}(T) = E(T^*)\setminus \GOOD_{T^*}(T)$.

The following lemma says that for every tree $\tilde{T}$ obtained
from $T$ by an SPR move using $S$, if $\tilde{T}$ has more bad
edges than $T$, then this was not the optimal SPR move using $S$.
Namely, there is a tree $T'$ which is also obtained from $T$ by an
SPR move using $S$, and $T'$ is such that $A(T')<A(\tilde{T})$.
(More precisely, each term in $A(T')$ is less than or equal to the
corresponding term in $A(\tilde{T})$ and there is a term in
$A(T')$ which is strictly smaller than the corresponding term in
$A(\tilde{T})$). Our proof has some similarity to those of Bruen
and Bryant \cite{BruenBryant} which connect the parsimony score of
a character to the minimum number of SPR transitions needed to
obtain the character.

\begin{lemma}
\label{lem:T'}
For every generating tree $T^*$ and all trees $T,\tilde{T}$ where $T$ and $\tilde{T}$ differ
by a prune-and-regraft of a subtree $S$ and such that
there exists
$f^*\in \BAD_{T^*}(\tilde{T})\setminus \BAD_{T^*}(T)$ the
following holds. There exists a tree $T'$ which differs from $T$ by a
prune-and-regraft of $S$ and  such that
${\rm cut}_R(T')\leq {\rm cut}_R(\tilde{T})$ for every partition $R$ realized by single edges in
$T^*$ and ${\rm cut}_R(T') < {\rm cut}_R(\tilde{T})$ for partition $R$ realized by $f^*$ in $T^*$.
\end{lemma}

\begin{proof}
Suppose an edge $f^*\in E(T^*)$
is good for $T$ and is bad for $\tilde{T}$.
Let $L_1,L_2$ be the partition of the leaves induced by $f^*$ in $T^*$.
Thus, in $T$, there is an edge $f=(v_1,v_2)$ which
partitions the leaves into $L_1$ and $L_2$.
See Figure~\ref{fig:setup} for an illustration of the setup.

\begin{figure}[htb]
\begin{center}
\includegraphics[type=pdf,ext=.pdf,read=.pdf,height=1.2in]{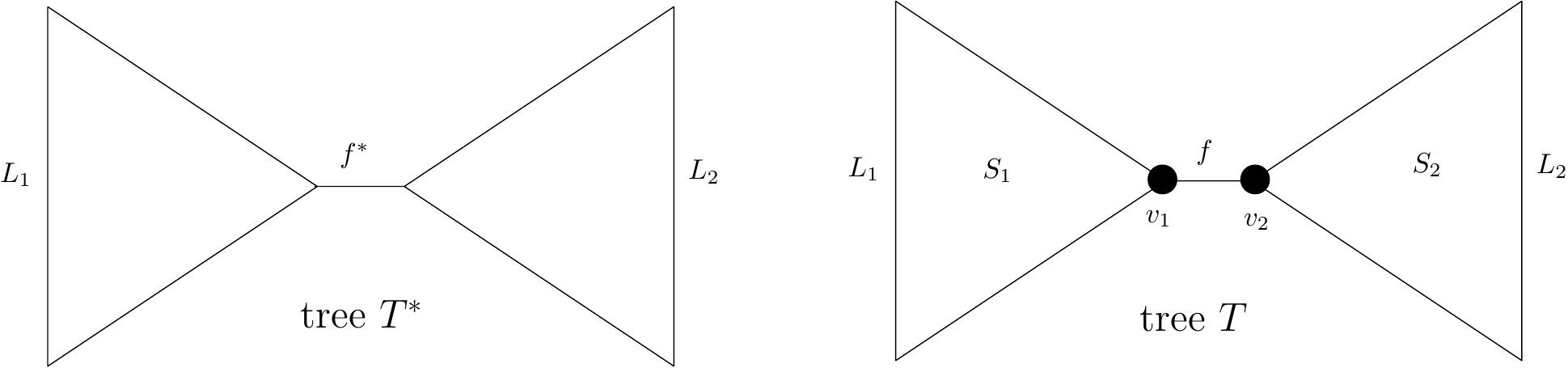}
\caption{Edge $f^*$ is good for $T$.} \label{fig:setup}
\end{center}
\end{figure}

Let $S_1$ denote the subtree ``hanging off'' of $v_1$ in $T$.
More precisely, after deleting $f$ from $T$, let $S_1$ be the subtree containing $v_1$.
Let $L_1$ denote the leaves in $S_1$.
Similarly, let
 $L_2$ denote the leaves and $S_2$ denote the subtree hanging off of $v_2$.
 Let $v$ denote the root of the subtree $S$.

First we claim that $f\not\in S$.  Suppose $f\in S$ and without loss of generality suppose
$S_1\subset S$.  See Figure \ref{fig:f-in-S} for an illustration of this case.
\begin{figure}[htb]
\begin{center}
\includegraphics[type=pdf,ext=.pdf,read=.pdf,height=1.5in]{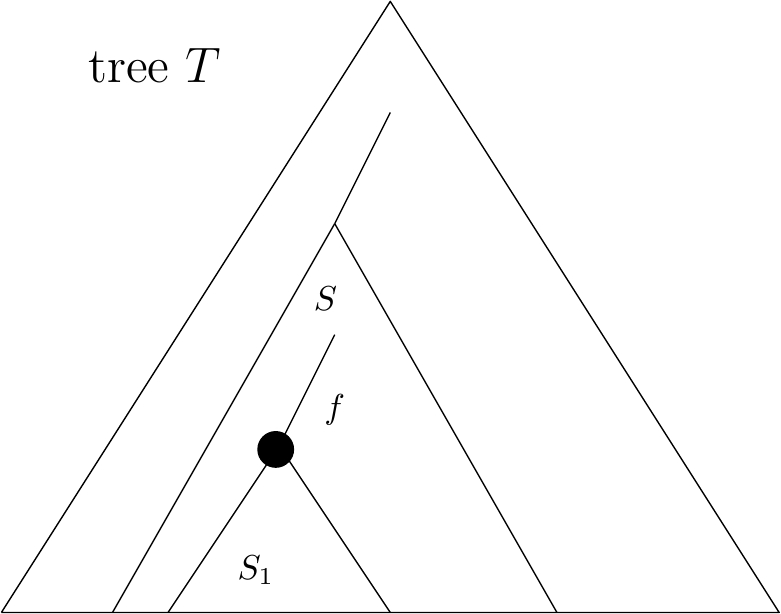}
\caption{Case when $f\in S$, this scenario cannot occur.}
\label{fig:f-in-S}
\end{center}
\end{figure}
Thus we must be grafting $S$ into an edge of $S_2\setminus S$.  After such a move,
the edge $f$ still separates $L_1$ and $L_2$, and thus $f^*$ is still good.
Therefore, $f\not\in S$.

From now on, we assume, without loss of generality, that $S\subset
S_1$ where $S\neq S_1$. see Figure \ref{fig:TT2}.
\begin{figure}[htb]
\begin{center}
\includegraphics[type=pdf,ext=.pdf,read=.pdf,height=1.3in]{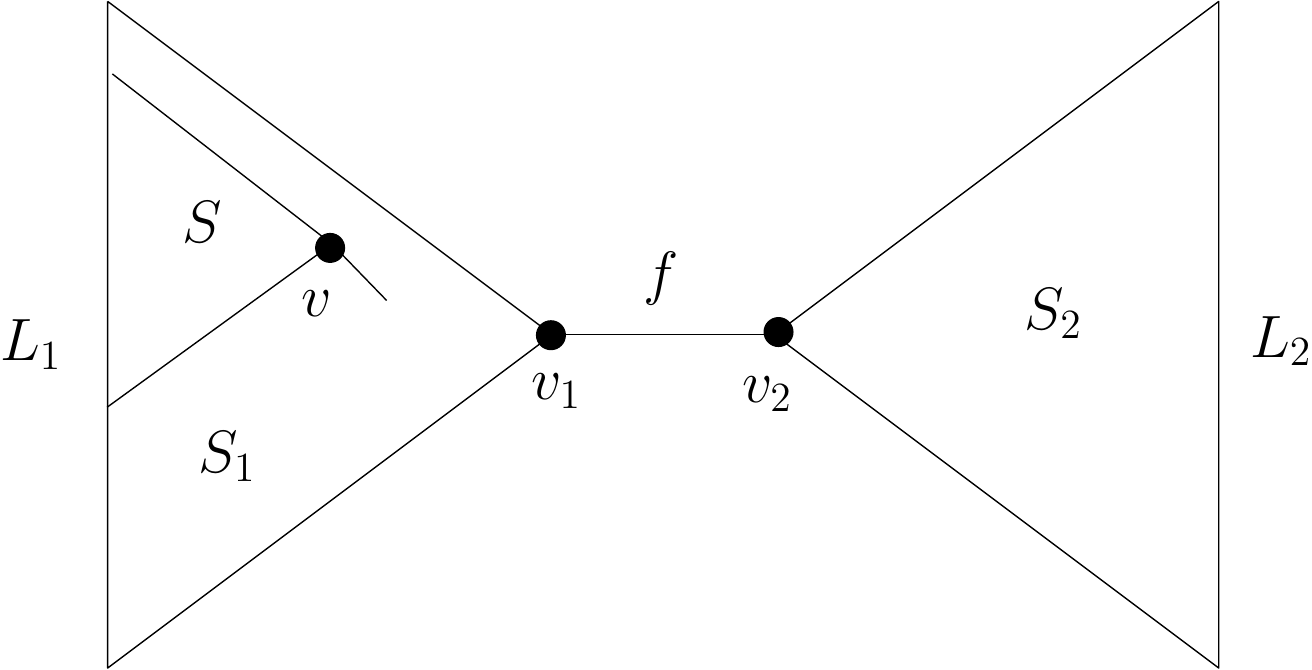}
\caption{Case when $f\not\in S$, this must be the scenario.} \label{fig:TT2}
\end{center}
\end{figure}
We construct the tree $T'$ by taking $T$, pruning $S$ and then regrafting $S$ along edge~$f$,
see Figure \ref{fig:T'}.

\begin{figure}[htb]
\begin{center}
\includegraphics[type=pdf,ext=.pdf,read=.pdf,height=1.5in]{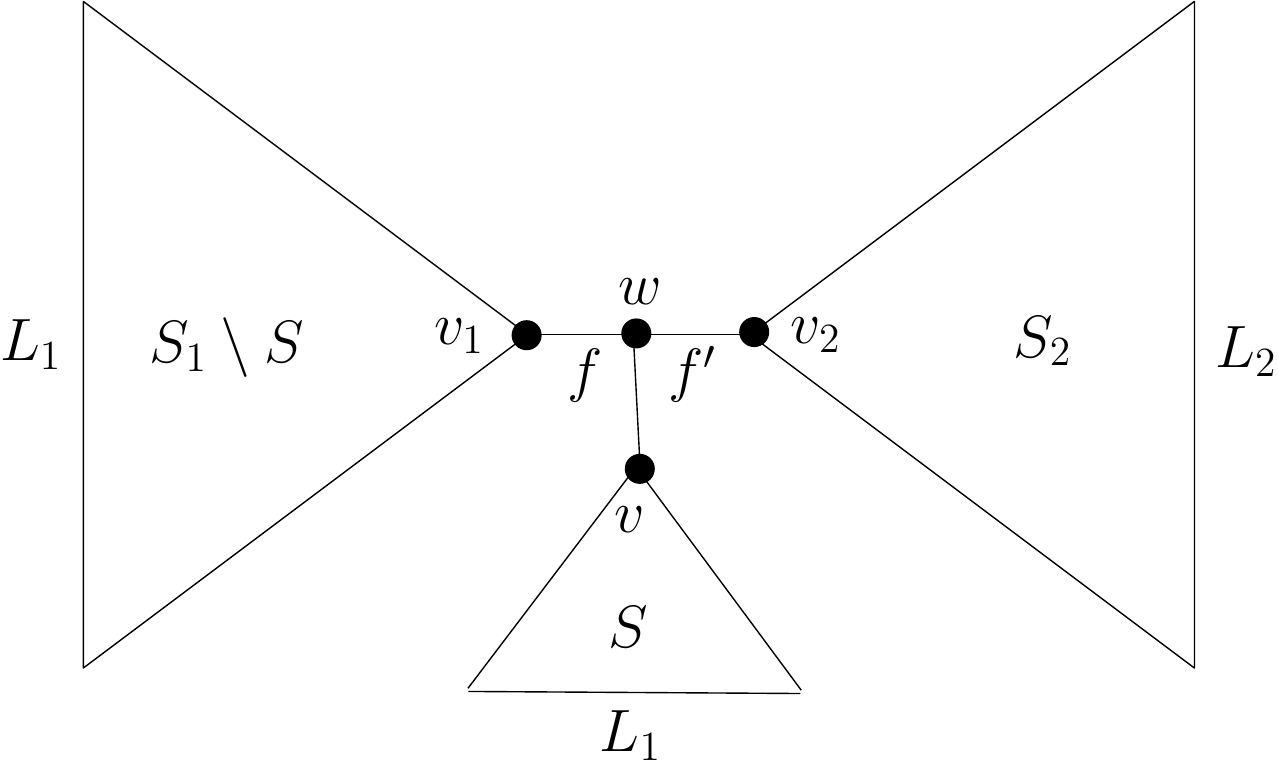}
\caption{Construction of the tree $T'$.}
\label{fig:T'}
\end{center}
\end{figure}
Note that $\tilde{T}$ is obtained from $T$ by regrafting $S$ onto an
edge in $S_2$ (otherwise $f^*$ would be good for $\tilde{T}$), see
Figure \ref{fig:Tww}.
\begin{figure}[htb]
\begin{center}
\includegraphics[type=pdf,ext=.pdf,read=.pdf,height=1.7in]{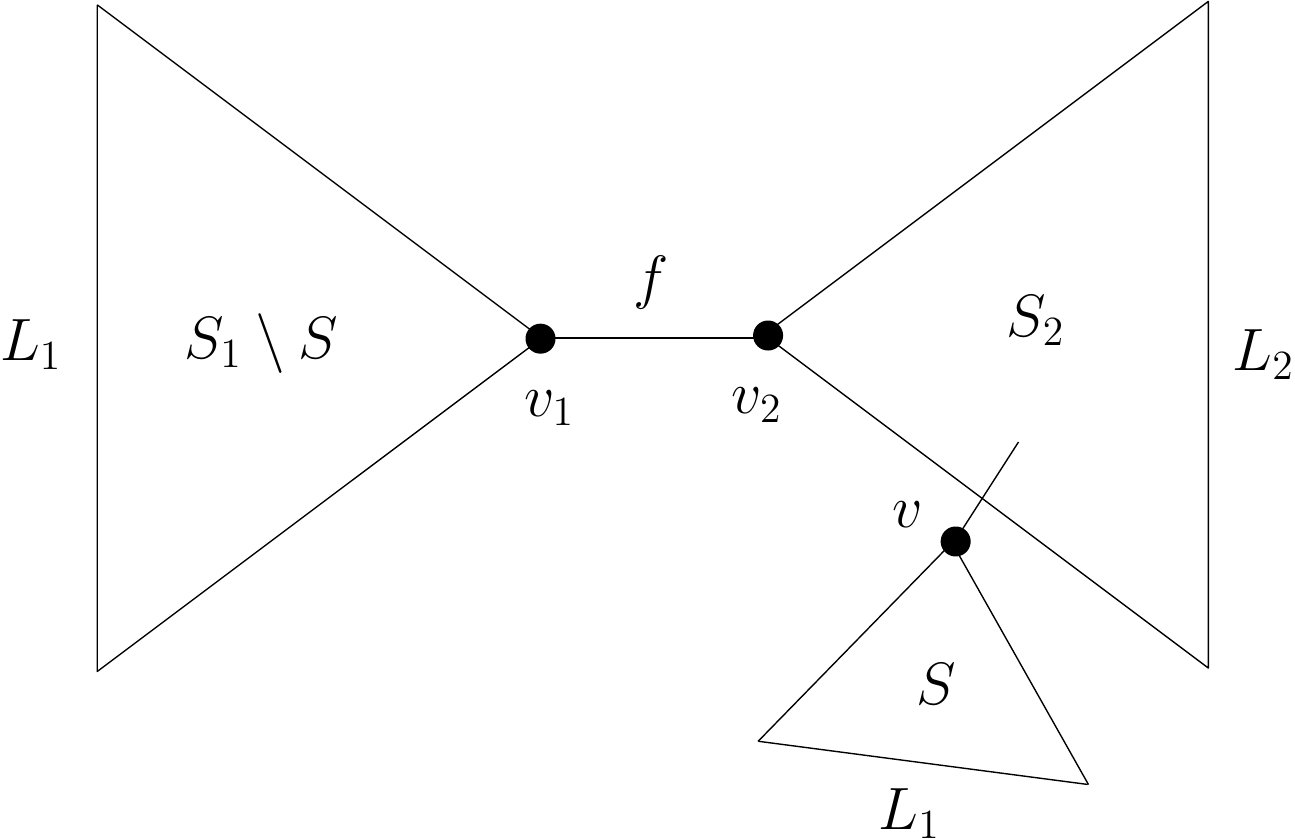}
\caption{In tree $\tilde{T}$, $S$ is regrafted into $S_2$.}
\label{fig:Tww}
\end{center}
\end{figure}

The following claim says that the tree $T'$ satisfies the conclusion of the lemma.

\begin{claim}
\label{claim:T'}
For every partition $R=(R_1,R_2)$ of leaves realized by edges of $T^*$,
it holds that:
\[ {\rm cut}_R(\tilde{T})\geq {\rm cut}_R(T'),
\]

Moreover, for the partition $R^*=(L_1,L_2)$ (corresponding
to $f^*$),
we have that:
\[{\rm cut}_{R^*}(\tilde{T}) > {\rm cut}_{R^*}(T').
\]
\end{claim}

The proof of the claim proceeds by constructing a cut in $T'$
realizing $(R_1,R_2)$ by a small modification of a cut in
$\tilde{T}$ realizing $(R_1,R_2)$.
Assuming the claim,
the proof of the lemma is now complete.
\end{proof}

We now prove the above claim.

\begin{proof}[Proof of Claim \ref{claim:T'}]

We continue using the setup and notation from the proof of Lemma
\ref{lem:T'} in
Section~\ref{sec:A-analysis}.

Recall, the claim says that
for every partition $R=(R_1,R_2)$ of leaves realized by edges of $T^*$
that ${\rm cut}_R(\tilde{T})\geq {\rm cut}_R(T')$ and for the partition $R^*=(L_1,L_2)$,
${\rm cut}_{R^*}(\tilde{T}) > {\rm cut}_{R^*}(T')$.

First we argue that ${\rm cut}_{R^*}(\tilde{T}) > {\rm cut}_{R^*}(T')$.  Note that
${\rm cut}_{R^*}(\tilde{T}) \geq 2$ since $f^*$ is bad for $\tilde{T}$.  On the other hand,
${\rm cut}_{R^*}(T')=1$ since cutting $f''$ separates $L_1$ and~$L_2$.
Now we just need to argue that ${\rm cut}_R(\tilde{T})\geq {\rm cut}_R(T')$.

Let $g^*\in E(T^*)$ be an edge in $T^*$.
Let $R=(R_1,R_2)$ be the corresponding partition in $T^*$.
Note that if $g^*$ is in the subtree with leaves $L_1$ then
\begin{equation}
\label{case1}
R_1\subseteq L_1 \mbox{ and } R_2\supseteq L_2.
\end{equation}
On the other hand, if $g^*$ is in the subtree with leaves $L_2$ then
\begin{equation}
\label{case2}
R_2\subseteq L_2 \mbox{ and } R_1\supseteq L_1.
\end{equation}

Consider a minimum cut $C\subset E(\tilde{T})$ that realizes $(R_1,R_2)$ in $\tilde{T}$
and amongst these minimum cuts is the one with the fewest number of edges in
subtrees $S_1\setminus S$ and $S$.

We claim that $v_1$ is reachable from a leaf of
$S_1\setminus S$ in $\tilde{T}\setminus C$ and that $v$ is reachable from a leaf of
$S$ in $\tilde{T}\setminus C$. Suppose that $v_1$ is not reachable from a leaf of $S_1\setminus S$.
Let $e'$ be the edge in $C\cap (S_1\setminus S)$ closest to $v_1$. We claim that
$C'=(C\setminus\{e'\})\cup\{f\}$ realizes $(R_1,R_2)$. If there was a pair of leaves
in $S_1\setminus S$ each in different $R_i$ that are connected in $\tilde{T}\setminus C'$
then by the choice of $e'$ one of those leaves would be connected to $v_1$ in $\tilde{T}\setminus C$, a contradiction with the assumption that $v_1$ is not
reachable from a leaf of $S_1\setminus S$ in $\tilde{T}\setminus C$.
Thus $R_1$ and $R_2$ are still separated in $S_1\setminus S$ in $\tilde{T}\setminus C'$;
$R_1$ and $R_2$ are still separated by $C'$ in $S_2\cup S$ since $C=C'$ in this subtree;
and $f\in C'$ ensures that pairs across $f$ are separated.
Note that $|C'|\leq |C|$ and $C'$ has fewer edges in $S_1\setminus S$ and $S$, a contradiction
with the choice of $C$. Thus $v_1$ is reachable from some leaf of $S_1\setminus S$ in
$\tilde{T}\setminus C$. The argument for $S$ and $v$ is the same.

Since a leaf of $S$ is reachable from $v$ in $\tilde{T}\setminus C$, then in other words a (non-empty) subset of
$R_1$ and/or $R_2$ are reachable from $v$.  Moreover, since $C$ realizes the partition $(R_1,R_2)$,
then a subset of only one of the $R_i$ is reachable from $v$ in $\tilde{T}\setminus C$;
we will say $v$ is of type $R_i$ to signify the $R_i$ reachable from $v$.
Analogously, we say $v_i$ is of type $R_i$ for the set reachable from $v_1$.

If $v$ and $v_1$ are of the same type $R_i$, let
\[ C' = \begin{cases}
  C &
  \mbox{if }
  f\notin C \\
  (C\setminus\{f\})\cup\{f'\} & \mbox{if } f\in C
  \end{cases}
  \]
  We claim $C'$ realizes $(R_1,R_2)$ in $T'$.
  To see this, note that if a path (between a pair of leaves) exists in $T'$ and
  does not exist in $\tilde{T}$ then it must include $w$ which is the new vertex in $T'$
  where $S$ is regrafted, see Figure \ref{fig:T'} for an illustration.
  Now we argue that such a path can not connect a leaf in $R_1$
  with a leaf in $R_2$ in $T'\setminus C'$.
  Note that only a subset of $R_i$ is reachable from $w$ in $T'\setminus C'$,
   since $w$ can reach
  the same set of vertices (outside of $S$) in $T'\setminus C'$ as $v_1$ does in $\tilde{T}\setminus C$,
  and only a subset of $R_i$
  is reachable from $v$ in $S\setminus C$.
Finally, since $|C'|=|C|$ we have that ${\rm cut}_R(T')\leq {\rm cut}_R(\tilde{T})$ which completes the proof in
  this case.

  Now suppose $v_1$ is of type $R_1$ and $v$ is of type $R_2$.  This
  means a leaf of $S$ is in $R_2$ and since $S\subset S_1$, it is also in $L_1$ and
  we are in case \eqref{case1}, thus $R_2\supseteq L_2$.
  Note $C$ has to separate $v_1$ from $L_2$ by some set of edges $Q\subseteq C$.
  Let $C'=(C\setminus Q)\cup\{f\}$.
  The new pairs of leaves that are connected
  in $T'\setminus C'$ (but not in $\tilde{T}\setminus C$)
  are either both from $L_2$
  and hence $R_2$, or are connected by a path that exists in $T'$ and
  does not exist in $\tilde{T}$.
  As in the previous case,
  if a path (between a pair of leaves) exists in $T'$ and
  does not exist in $\tilde{T}$ then it must include $w$ which is the new vertex in $T'$
  where $S$ is regrafted.
  Note that $w$ is disconnected from $S_1\setminus S$ in $T'\setminus C'$ (since $f\in C'$).
  The leaves of $S_2$ are from $R_2$  and the leaves of $S$ reachable from $v$ in $S\setminus C$
  are also from $R_2$.   Therefore the new paths do not connect leaves of $R_1$ and $R_2$.
  This completes this case since $|C'|\leq |C|$.

  Finally, suppose $v_1$ is of type $R_2$ and $v$ is of type $R_1$.
  In this case a leaf of $S_1\setminus S$ is in $R_2$ and is also in $L_1$ and
 therefore we are again in case \eqref{case1}, thus $R_2\supseteq L_2$.
   Note $C$ has to separate $v$ from $L_2$ by some set of edges $Q\subseteq C$.
  Let $C'=(C\setminus Q)\cup \{w,v\}$.
Once again, the new pairs of leaves that are connected
  in $T'\setminus C'$ (but not in $\tilde{T}\setminus C$)
  are either both from $L_2$
  and hence $R_2$, or are connected by a path that exists in $T'$ and
  does not exist in $\tilde{T}$.
  Note that $w$ is disconnected from the leaves of $S$ in $T'\setminus C'$.
  The leaves of $S_2$ are from $R_2$  and the leaves of $S_1\setminus S$
  reachable from $v_1$ are also from $R_2$.
  Therefore the new paths in $T'\setminus C'$ do not connect leaves of $R_1$ and $R_2$.
  This completes this case since $|C'|\leq |C|$.

This completes the proof of the claim.
\end{proof}

Using Lemma \ref{lem:T'},
we will prove that for every subtree $S$ the optimal SPR move using
$S$ does not increase the number of bad edges, and there is a
subtree $S$ where the optimal SPR move using $S$ decreases the
number of bad edges.  It will then be
straightforward to prove rapid mixing by analyzing the time until
the number of bad edges is zero, and hence we have reached $T^*$.

\begin{lemma} \label{lem:increase}
For all trees $T^*$, every choice of parameters $\alpha:E(T^*)\rightarrow\reals^+$,
for all trees $T\neq T^*$ the following holds, where $A=A_{T^*,\vec{\alpha}}$ is
defined in \eqref{def:A}.
\begin{enumerate}
\item
For any subtree $S$ of $T$ the following holds.
Let $\Tmin$ be any tree which minimizes $A(\Tmin)$
amongst the SPR neighbors of $T$ which differ by a prune-and-regraft of $S$.
Then,
\begin{equation}
\label{eq:part1}
\BAD_{T^*}(\Tmin) \subseteq \BAD_{T^*}(T).
\end{equation}
\item
There exists a subtree $S$ of $T$ where the following holds.
Let $\Tmin$ be any tree which minimizes $A(\Tmin)$
amongst the SPR neighbors of $T$ which differ by a prune-and-regraft of $S$.
Then,
\begin{equation}
\label{eq:part2}
\BAD_{T^*}(\Tmin) \subsetneq \BAD_{T^*}(T).
\end{equation}
\end{enumerate}
\end{lemma}

Part 1 of Lemma \ref{lem:increase} follows immediately from
Lemma \ref{lem:T'}.  To prove part 2 we choose a particular `minimal'
subtree $S$.  Roughly speaking, we consider the bad edge $f^*$ that is
closest to the leaves in $T^*$, and take the subtree $S$ hanging off of $f^*$.

\begin{proof}[Proof of Lemma \ref{lem:increase}]

If \eqref{eq:part1} is violated
then there exists $f^*\in \BAD_{T^*}(\Tmin) \setminus\BAD_{T^*}(T)$,
and hence by Lemma \ref{lem:T'}, there exists $T'$ (which differs
from $T$ and $\Tmin$ by a prune-and-regraft of $S$)  such that
no cuts increased in size and the cut corresponding to $f^*$ is smaller.
Therefore, $A(T')<A(\Tmin)$, contradicting the choice of $\Tmin$.
Therefore, part 1 holds.

We now prove part 2.  We first claim that there is an SPR move that decreases the
number of bad edges.

\begin{claim}
\label{claim:fewer-bad}
For every tree $T$, there is an SPR move resulting in a tree $T'$ where
\begin{equation}
\label{eq:fewer-bad}
\BAD_{T^*}(T') \subsetneq \BAD_{T^*}(T).
\end{equation}
\end{claim}

Now we argue that part 2 of the lemma follows from the above claim and
part~1.  We then go back to prove the claim.

Consider a subtree $S$ of $T$.  Let $N_S(T)$ denote those trees obtainable from $T$
by a prune-and-regraft of $S$.  Note, for any $T'\in N_S(T)$, we have that
$N_S(T')=N_S(T)$, since when we prune $S$ from $T$ and $T'$ we have the
same subtree remaining.

Let $T'$ denote the neighboring tree from Claim \ref{claim:fewer-bad}
with fewer bad edges, and let $S$ denote the subtree where $T'\in N_S(T)$.
Let $T_{\min}$ denote the tree in $N_S(T)$ which minimizes $A(T_{\min})$.
As noted above we must have that $N_S(T')=N_S(T)$.
Thus, $T_{\min}$ is also the neighbor of $T'$ that minimizes $A(T_{\min})$.
Therefore, we can apply part 1 of the lemma for tree $T'$ and subtree $S$, and conclude that
$\BAD_{T^*}(T_{\min}) \subseteq \BAD_{T^*}(T')$.
Combined with \eqref{eq:fewer-bad} we then have that:
\[
\BAD_{T^*}(T_{\min}) \subsetneq
\BAD_{T^*}(T),
\]
which proves part 2.
\end{proof}

We now prove Claim \ref{claim:fewer-bad}.

\begin{proof}[Proof of Claim \ref{claim:fewer-bad}]

Let $f^*$ in $T^*$ be an edge in $\BAD_{T^*}(T)$ that is ``closest'' to the leaves
in the following precise sense.  Say $f^*$ joins subtrees $S^*$ and $Z^*$
in $T^*$ where the number of vertices in $S^*$ is at most the number of
vertices in $Z^*$.  Then we say the distance of $f^*$ to the leaves is
the number of vertices of $S^*$.

Note, by the choice of $f^*$,
$S^*$  contains no bad edges for $T$.  First, note that $S^*$ must contain at least
two leaves because, in any tree, any single leaf can be separated
from the rest of the leaves by deleting one edge (which would contradict that $f^*$ is bad).
Let $S^*_1$ and $S^*_2$ denote the two subtrees of $S^*$ hanging from the root of $S^*$ in $T^*$.
Both $S^*_1$ and $S^*_2$ must exist since $S^*$ contains at least two leaves.

Let $L_1$ and $L_2$ denote the leaves in $S^*_1$ and $S^*_2$, respectively.
Since $f^*$ is the closest bad edge to the leaves,
there is a subtree $S_1$ in $T$ whose leaves are $L_1$, and also
a subtree $S_2$ whose leaves are $L_2$.
Moreover, by induction, $S_1=S^*_1$ and $S_2=S^*_2$.
In $T$, by pruning $S_2$ and then regrafting along the edge incident to $S_1$
we obtain a copy of $S^*$ in $T$.
See Figure
\ref{fig:claim} for an illustration.
Let $T'$ be the tree resulting from this SPR move.
Note, $f^*$ is now a good edge in $T'$.

\begin{figure}[htb]
\begin{center}
\includegraphics[type=pdf,ext=.pdf,read=.pdf,height=1.5in]{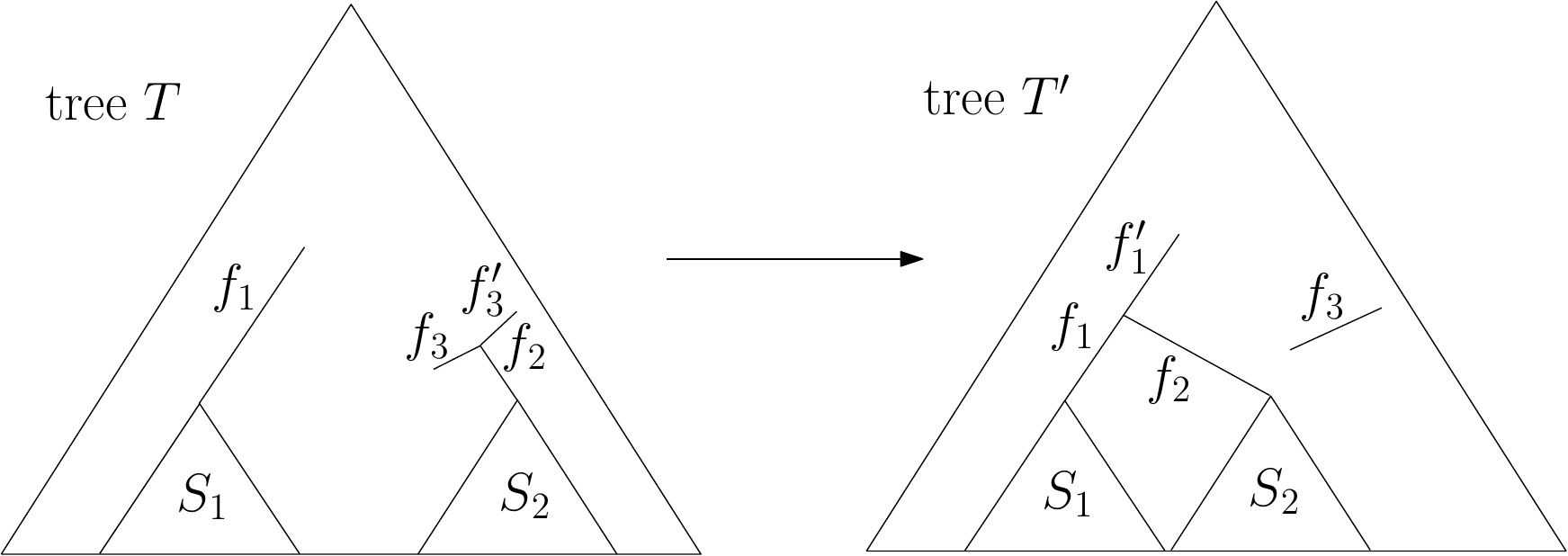}
\caption{Construction of the tree $T'$ with fewer bad edges.}
\label{fig:claim}
\end{center}
\end{figure}

It remains to argue that other edges of $T^*$ did not change from good for $T$ to bad for $T'$.
Note, edges in $S^*_1$ and $S^*_2$ remain good in $T'$ since they are realizable
in $S_1$ and $S_2$, respectively.    Consider an edge $e^*$ of $T^*$ where $e^*\notin S^*$.   Let $(R_1,R_2)$ be the partition
of the leaves realized by $e^*$ in $T^*$. Note that $L_1,L_2$ are
in the same partition, since tree $S^*$ is not cut by $e^*$.
Let $g$ be the edge in $T$ that realizes $(R_1,R_2)$.
After pruning-and-regrafting $S_2$ (to form $T'$), $g$ still realizes the partition
$(R_1,R_2)$ since $L_1$ and $L_2$ are in the same partition.
 Hence, $e^*$ is still good for $T'$.
Therefore, $\GOOD_{T^*}(T') \supseteq \GOOD_{T^*}(T)\cup\{f^*\}$, which
completes the proof of the claim.
\end{proof}

Finally, we prove that when the number of
bad edges increases then $A(T)$ also increases by
a significant amount.  As a consequence, in our analysis
of the Markov chain, by taking
$\eps$ sufficiently small, we can focus on how a
transition changes $A(T)$ and hence on the change in
the number of bad edges.

\begin{lemma}
\label{lem:gap}
For any trees $T$ and $T'$ which differ by one SPR move, if
$|\BAD_{T^*}(T')| > |\BAD_{T^*}(T)|$
then
\[
A(T') \geq A(\Tmin) + \alpha_{\min}.
\]
\end{lemma}

\begin{proof}
Let $S$ be the subtree used to move between $T$ and $T'$.
Let $N_S(T)$ denote those trees obtainable from $T$ by
a prune-and-regraft of $S$.   Note $N_S(T) = N_S(T')$.

Consider $\Tmin$ which minimizes $A(\Tmin)$ amongst
 the SPR neighbors of $T$ which differ by a prune-and-regraft of $S$.
Since $N_S(T')=N_S(T')$ then $\Tmin$ is also the neighbor of $T'$
that minimizes $A(\Tmin)$.
Fix $e\in \BAD_{T^*}(T')\setminus\BAD_{T^*}(T)$.
By Part 1 of Lemma \ref{lem:increase},
\[  \BAD_{T^*}(\Tmin) \subseteq \BAD_{T^*}(T)\cap\BAD_{T^*}(T').
\]

\end{proof}

\subsection{Proof of Rapid Mixing: Theorem \ref{thm:main}}
\label{sec:rapid-mixing}

The proof of our main theorem now follows from a straightforward argument.
We show that the Heat Bath SPR Markov Chain behaves like a local
search algorithm, and then a simple coupling argument
gives the mixing result.

\begin{proof}[Proof of Theorem \ref{thm:main}]
Let $\mathcal{T}$ denote the space of phylogenetic trees on $n$ taxa.
For a tree $T\in\mathcal{T}$, and subtree $S$ of $T$,
let  $N_S(T)$ denote those trees obtainable from $T$
by pruning-and-regrafting $S$.

Let $C$ be the constant in the $O(\cdot)$ notation of \eqref{eee2}
for the chosen $\alpha_{\min}$ and $\alpha_{\max}~=~1$.
By choosing $\eps_0$ (note that $\eps_0$ is an upper bound on $\eps$)
sufficiently small then for every tree~$T$, in \eqref{eee2},
the $C\eps\ln\ln(1/\eps)$ is smaller than
$|\alpha_{\min}(\eps\ln{\eps})/10|$ and therefore:
\begin{equation}\label{L:first-order}
{\mathcal L}_{T}(\mu^*) =
{\cal E}(\pi^*) + (A(T)+\delta_T)\eps\ln\eps,
\end{equation}
for some $|\delta_T|<\alpha_{\min}/10$.

Fix a tree $T\neq T^*$ and a subtree $S$ of $T$.
By Lemma \ref{lem:gap} and \eqref{L:first-order},
for every $T'\in N_S(T)$
where $|\BAD_{T^*}(T')|>|\BAD_{T^*}(T)|$ we have:
\[
{\mathcal L}_{T'}(\mu^*)  <
{\mathcal L}_{\Tmin}(\mu^*) - (9/10)\alpha_{\min}\eps\ln (1/\eps).
\]

For a character $\sigma\in\Omega^n$, let
$D(\sigma) = |\{i: D_i = \sigma\}|$.
A straightforward application of Hoeffding's inequality
\cite{Hoeffding}
and a union bound
over $\sigma\in\Omega^n$ implies,
for all $\delta>0$:
\[
\Prob{ \mbox{for all }\sigma\in\Omega^n, |D(\sigma) - \mu^*(\sigma)N|
\le \delta N}  \geq 1 -2\cdot 4^n\exp(-2\delta^2N).
\]
Let $q_{\min} = \min_{i,j\in\Omega: i\neq j} Q_{i,j}$ denote a lower-bound
on the off-diagonal entries in the rate matrix.
For $\eps_0$ sufficiently
small, every labeling of the leaves has probability at least $\eps^{2n}$; this follows
from the fact that for every edge, every transition has probability
$\Omega(\eps)$, see \eqref{P:ij} for a precise statement.
Hence, by
choosing $\eps_0$ sufficiently small (relative to
$\alpha_{\min}, q_{\min}$ and the constant in the error term of \eqref{P:ij}),
then for all $\sigma\in\Omega^n$,
$\mu^*(\sigma)\geq \eps^{2n}$.
Let
\[ \delta = \alpha_{\min}\eps\ln(1/\eps) /(20\cdot 4^nn\ln{\eps}).\]
Then, for $\mathbf{D}\sim\mu^*$, \[
{\mathcal L}_{T'}(\mathbf{D})  <
{\mathcal L}_{\Tmin}(\mathbf{D}) - (7/10)\alpha_{\min}\eps\ln(1/\eps) N,
\]
with probability $\geq 1-\exp(-\sqrt{N})$ for $N$ sufficiently large.
The probability of moving from $T$ to $T'$ after choosing $S$ in step 1
is at most:
\begin{equation}
\label{eq:lik-exp}
\frac{\exp({\mathcal L}_{T'}(\mathbf{D}))}
{\exp({\mathcal L}_{\Tmin}(\mathbf{D}))} < \exp(-(7/10)\alpha_{\min}\eps \ln(1/\eps)  N)  < \exp(-10n)
\end{equation}
for $N$ sufficiently large. Therefore, with probability $\geq
1-4n\exp(-10n)$, the chain will move from $T$ to some $\Tmin$
(where $\Tmin$ is a tree that can be obtained from $T$ by an SPR move and such that it minimizes
$A(T_{\min})$),
and thus by part 1 of Lemma \ref{lem:increase}
the number of bad edges will not increase.  Moreover, if we choose the subtree
$S$ satisfying part 2 of Lemma \ref{lem:increase} then the
number of bad edges will
decrease.  Hence, with probability  $\geq 1/(4n) -4n\exp(-10n)\geq 1/(5n)$ the number of bad
edges decreases by at least one.  In expectation, after $\leq 5n$ steps of the chain
the number of bad edges will be zero, in which case we have reached $T^*$.
By Markov's inequality, with probability $\geq 9/10$
after $50n$ steps we reach $T^*$.
Once we reach $T^*$ the probability of moving to a different tree within $50n$ steps is at most
$50n(4n)^2\exp(-10n)<1/100$.  Hence the claimed mixing time follows by an elementary
coupling argument
(c.f., \cite{Peres:book}
for an introduction to the coupling technique)
since from any pair of initial trees, both chains
(run independently) reach $T^*$ at time $50n$ with probability  $\geq 1-1/2\eee$.
\end{proof}

\section{Bayesian Inference}
\label{sec:posterior}

The goal is often to randomly sample from the posterior distribution over
trees.  To do this, we consider a Markov chain whose stationary
distribution is the posterior distribution and analyze the chain's
mixing time, which is a measure of the convergence time of the
chain to its stationary distribution. Let $\Phi(T,Q,\vt)$ denote
a prior density where
\[
\sum_{T}\int_{Q\in\Q}\int_{\vt} \Phi(T,Q,\vt)d\vt dQ = 1.
\]
Our results extend to priors that are lower bounded by some
$\delta>0$ as in Mossel and Vigoda
\cite{MV2}.
In particular, for
all trees $T$ and all branch lengths~$\vt$ where $\vt_e\leq
t_0$ for all edges $e$, we require $\Phi(T,Q^*,\vt)\geq
\delta$. We refer to these priors as $(\delta,t_0)$-regular
priors.

Applying Bayes law we get the posterior distribution:
\[
\ProbCond{T,Q,\vt}{\Dvec}
 =
\frac{ \mu_{T,Q,\vt}(\Dvec)\Phi(T,Q,\vt) } {\Prob{\Dvec} }, \]
where
\[
\Prob{\Dvec}=
 \sum_{T'} \int_{Q'\in\Q}\int_{\vt'}
  \mu_{T',Q',\vt'}(\Dvec)\Phi(T',Q',\vt')d\vt 'dQ.
  \]

Each tree $T$ then has a posterior weight
\begin{equation}
\label{eq:posterior-dist}
w(T) = \int_{Q\in\Q}\int_{\vt}
   \mu_{T,Q,\vt}(\Dvec)\Phi(T,Q,\vt)d\vt dQ.
\end{equation}
Finally, the posterior distribution $\mu$ on trees is defined as
$\mu(T) = w(T)/\sum_{T'} w(T')$.

\subsection{Extension of Theorem \ref{thm:main} to Sampling the Posterior}

To sample from the posterior distribution,
the Markov chain is defined as in Section
\ref{sec:SPR}
except that
in step 2 the weight $w(e^*)$ is now set as $w(T^*)$ defined in
 \eqref{eq:posterior-dist}.  This ensures that the Markov chain is
 reversible with respect to the posterior distribution, and hence this is
 the unique stationary distribution.

 Theorem \ref{thm:main} then extends to hold for any priors which
 are $(\delta,2\eps_0)$-regular
 The proof easily extends to this case in the following manner.

 In particular, we need
to modify the statement of Lemma \ref{lem:opt-weights}
so that, for any tree $T$,
\eqref{eee2} is achieved for $Q=Q^*$ and
for every set of branch lengths $\vt$ where $\vt_e\in (\eps/2,2\eps)$
for all edges $e$.
Then we can use the same proof as Lemma 21 in
Mossel and Vigoda
\cite{MV2}
to get an analog of \eqref{eq:lik-exp} to hold for the posterior weights defined
in \eqref{eq:posterior-dist} in place of the maximum likelihood function
$\exp(\mathcal{L}(\vec{D}))$,
and the remainder of the proof of Theorem \ref{thm:main} remains the same.

\section{Discussion}
\label{remarks}

{\bf NNI Transitions:}  In a NNI transition, an internal edge $e$ is chosen.
Since internal vertices have degree three, there are four subtrees hanging off of~$e$.
There are three possible ways of attaching these four subtrees to $e$, and an NNI transition
moves to one of these rearrangements.
There are trees $T$ (different from the generating tree $T^*$)
where no NNI neighbor (strictly) improves $A(T)$; moreover, there are cases
where there is also no improvement in the next term of \eqref{eee2}.
We are uncertain as to whether Theorem \ref{thm:main} holds for a Markov chain
based on NNI transitions.  It would be especially intriguing if there are cases where
chains based on NNI transitions are slow to converge (so-called torpidly mixing),
whereas a chain based on SPR transitions is provably fast to converge (rapidly mixing).

{\bf Possible Future Work:}
There are now several works with
proofs of convergence of MCMC algorithms for phylogenetic reconstruction in
certain settings --
rapid mixing results in this paper and torpid mixing results in
Mossel and Vigoda \cite{MV1,MV2}
and
\v{S}tefankovi\v{c} and Vigoda
\cite{SV1,SV2}).
All of these results require
that the branch lengths are sufficiently small so that
only the dominant terms of the likelihood function need to be considered.
A natural avenue for extending this paper, is to allow arbitrary branch lengths
on the terminal edges.

{\bf Rapid or Torpid Mixing for General Pure Distributions:}
The most tantalizing question to the authors is whether there exists a pure distribution
(i.\,e., a single generating tree as in the setting of this paper) where Markov chains
based on all of the natural transitions (e.\,g., NNI, SPR and TBR
transitions)
are slow to converge to the stationary distribution (in other words, they are torpidly mixing).
We expect simulations can be quite useful for finding such a bad example, if one
exists; in fact, our previous work
\cite{SV1,SV2}
on this topic was inspired by some
intriguing findings from some simple simulations.

\section{Acknowledgements}

We are thankful to the anonymous referees for many useful
comments on the manuscript.

\end{document}